\documentclass[smallextended]{svjour3}       
\smartqed  
\usepackage{graphicx}

\usepackage{mathptmx}      
\usepackage{graphics} 
\usepackage{epsfig} 
\usepackage{mathptmx} 
\usepackage{times} 
\usepackage{amsmath} 
\usepackage{amssymb}  

\usepackage{amsthm}
\usepackage{comment}
\usepackage{color}
\usepackage{xcolor}
\usepackage{graphicx}
\usepackage{float}
\usepackage[]{algorithmic}
\usepackage[justification=centering]{caption}

\usepackage{lipsum}

\newtheorem{algorithm}{Procedure}

\floatname{algorithm}{Procedure}

\journalname{Discrete Event Dynamic Systems}

\newcommand{\MR}[1]{\textcolor{blue}{MR: #1}}
\begin{document}

\title{Synthesis of Covert Actuator Attackers for Free 
%
}


\author{
Liyong Lin  \and Yuting Zhu \and Rong Su 
}

\authorrunning{L. Lin, Y. Zhu,  R. Su} 

\institute{Liyong Lin, Yuting Zhu, Rong Su  \at
              Electrical and Electronic Engineering, Nanyang Technological University, Singapore \\
\email{liyong.lin@ntu.edu.sg, yuting002@e.ntu.edu.sg, rsu@ntu.edu.sg} 
          }
          
\vspace{-20pt}
\date{Received: date / Accepted: date}

\vspace{-10pt}

\maketitle

\begin{abstract}
In this paper, we shall formulate and address a problem of covert actuator attacker synthesis for  cyber-physical systems that are modeled by discrete-event systems. We assume the actuator attacker  partially observes the execution of the closed-loop system and is able to modify each control command issued by the supervisor on a specified attackable subset of  controllable events. We   provide  straightforward but in general exponential-time reductions, due to the use of subset construction procedure, from the covert actuator attacker synthesis problems to the  Ramadge-Wonham supervisor synthesis problems. It then follows that it is possible to use the many techniques and tools already developed for solving the  supervisor synthesis problem to solve the covert actuator attacker synthesis problem for free. In particular, we show that, if the attacker cannot attack unobservable events to the supervisor, then the reductions can be carried out in polynomial time.  We also provide a brief discussion on some other conditions under which the exponential blowup in state size can be avoided. Finally, we show how the reduction based synthesis procedure can be extended for the synthesis of  successful covert actuator attackers that  also eavesdrop the control commands issued by the supervisor.   
\end{abstract}

\makeatletter
\let\origsection\section
\renewcommand\section{\@ifstar{\starsection}{\nostarsection}}

\newcommand\nostarsection[1]
{\sectionprelude\origsection{#1}\sectionpostlude}

\newcommand\starsection[1]
{\sectionprelude\origsection*{#1}\sectionpostlude}

\newcommand\sectionprelude{%
  \vspace{0.1em}
}

\newcommand\sectionpostlude{%
  \vspace{0.1em}
}
\makeatother

\vspace{-10pt}
\section{Introduction}
\vspace{-4pt}

Recently, security of cyber-physical systems has drawn much research interest within the discrete-event systems and formal methods community~\cite{CarvalhoEnablementAttacks},~\cite{Carvalho2018},~\cite{Su2018},~\cite{Goes2017},~\cite{Lanotte2017},~\cite{Jones2014},~\cite{WTH17},\\~\cite{LACM17},~\cite{Lima2018},~\cite{K2016},~\cite{R17},~\cite{Lin2018},~\cite{Zhu2018},~\cite{LZS19},~\cite{WBP19},~\cite{YL19},~\cite{HMR18}. For a recent survey on the discrete-event systems based  approach for the security of  cyber-physical systems, the reader is referred to~\cite{Aida19}.
 In this paper, we shall focus on  discrete-event systems as our model of cyber-physical systems and consider the problem of attacker synthesis, as a major step towards solving the resilient supervisor synthesis problem~\cite{LZS19}. In particular, we consider the synthesis of a successful actuator attacker~\cite{Lin2018},~\cite{LZS19}. We assume the goal of the actuator attacker is to remain covert in the process of attacking the closed-loop systems until it causes damages, much in parallel to the framework of~\cite{Su2018}. An actuator attacker is  said to be successful if it achieves its goal. Correspondingly, a supervisor is said to be resilient if a successful actuator attacker does not exist. We also consider the synthesis of (successful) covert actuator attackers that can eavesdrop the control commands issued by the supervisor, extending our recent work~\cite{Lin2018}. 
 
Some of the existing works  consider the problem of deception attack~\cite{Su2018},~\cite{Goes2017},~\cite{WTH17}, where the attacker would accomplish the attack goal by altering the sensor readings. There are also works that deal with actuator attacks as well as sensor attacks~\cite{Carvalho2018},~\cite{LACM17},~\cite{Lima2018},\\~\cite{WBP19}. We provide some comparisons between the existing works and this work in the following.
\begin{enumerate}
    \item The works of~\cite{Carvalho2018},~\cite{LACM17} and~\cite{Lima2018}  do not require the  attackers to possess any knowledge of the system's models\footnote{There is in fact an exception. In sensor insertion scenario,~\cite{Carvalho2018} assumes the attacker knows the model of the supervisor.} and, in particular, they only consider the reachability of damages as the attackers' goal; consequently, they only need to consider the ``worst case" attacker~\cite{Carvalho2018} that performs enablement attacks wherever it is possible, which makes the problem of attacker synthesis  irrelevant in their problem setup. In contrast, we consider those attackers that know the models of the system, including the plant's and the supervisor's models. This allows the attacker to make informed attack decisions and also makes it possible to model different goals and constraints of the attackers. In particular, this does allow us to model the synthesis of covert attackers~\cite{Su2018},~\cite{Lin2018},~\cite{LZS19}. Indeed, in this paper, we consider the synthesis of covert attackers with both  damage-nonblocking goal and damage-reachable goal. In addition, compared with~\cite{Carvalho2018},~\cite{LACM17} and~\cite{Lima2018}, we allow the attacker and the supervisor to have different observation capabilities. While disablement attack is considered to be   not useful in~\cite{Carvalho2018},~\cite{LACM17} and~\cite{Lima2018}, it is often crucial for the success of covert attackers, when combined with enablement attacks.   
    \item The only means of system protection in the setup of~\cite{Carvalho2018} and~\cite{LACM17} is by making use of an intrusion detection module for detecting the presence of an attacker, along with a fixed supervisor that may not have been properly designed against the attackers. 
 Their setup does not consider the possibility of modifying the supervisor to render attack  impossible or harder. In comparison, the resilient supervisor synthesis based approach in~\cite{Su2018},~\cite{LZS19} allows and explicitly considers the synthesis of a  supervisor that can prevent the damages caused by attackers,  by making attack impossible or by making attack futile\footnote{The supervisor could be designed in such a way that either every disabled event  is non-attackable, which makes (enablement) attack impossible, or every disabled attackable event cannot bring the system to damage, which makes attack futile.}, even in the absence of a monitor, or by discouraging attack attempt that may cause the violation of covertness, by using a monitor~\cite{Lin2018}. The feature that 
 the presence of a monitor can discourage attack attempt  is not available in~\cite{Carvalho2018} and~\cite{LACM17}.
 The work of~\cite{WBP19}  attempts to  generalize~\cite{Su2018} and uses non-deterministic finite state transducers as the unifying models to address the security against combined actuator attacks and sensor attacks. However, it only considers  resilient supervisor  verification  upon newly defined controllability condition. The problem of synthesis of resilient supervisor remains unsolved if the controllability condition is not satisfied. Similarly,~\cite{Carvalho2018} and~\cite{LACM17} only address the problem of verification, instead of synthesis, for security. We note that solving the problem of attacker synthesis in this work naturally solves the resilient supervisor verification problem and can be used for solving the resilient supervisor synthesis problem as well, as explained in~\cite{LZS19}.  While~\cite{Lima2018} also studies undetectable attacks (in addition to detectable attacks), they consider the undetectability of the ``worst case" attacker and the UNA-security verification condition~\cite{Lima2018} under which security against undetectable ``worst case" attacker can be established. We start from covertness goal and synthesize covert attackers fulfilling  damage-infliction goals. Again, \cite{Lima2018} only studies the problem of verification, instead of synthesis, for security.
 While resilience  against risky attackers of~\cite{Carvalho2018},~\cite{LACM17} and~\cite{Lima2018} implies resilience against covert attackers, if the system designer has the additional information that any potential attacker intends to remain covert, this information can be used and often allows a larger design space of resilient supervisors for the system designer.
 \item In~\cite{Lin2018},~\cite{Zhu2018} and~\cite{LZS19}, we have considered  partially-observing attackers that also eavesdrop the control commands issued by the supervisor. Given the same observation capability over the plant events, control command eavesdropping attackers are more powerful than the counterparts that do not eavesdrop the control commands.  In~\cite{HMR18}, a notion of a powerful  attacker is defined for passive attackers; a passive attacker is considered to be powerful in~\cite{HMR18} if it can directly query some information about the current state of the plant. A control command eavesdropping attacker in~\cite{Lin2018} can be considered as a powerful attacker in the sense of~\cite{HMR18}. In particular, the supervisor's control command does encode some knowledge of the partially-observing supervisor on the execution of the plant. We consider the synthesis of both types of covert actuator  attackers and develop reduction based synthesis procedures. Compared with~\cite{Lin2018} that also addresses the covert actuator attacker synthesis problem, we do not require normality $\Sigma_c \subseteq \Sigma_o, \Sigma_{c, A} \subseteq \Sigma_{o, A}$ and the  synthesis procedure is much simpler.
\end{enumerate}
The paper is organized as follows: Section~\ref{sec: prel} is devoted to the preliminaries. Then, in Section~\ref{section:SS}, we shall provide a detailed explanation about the system setup, including a formalization of the supervisor, the non-eavesdropping actuator attacker, the damage automaton, the monitoring mechanism, the attacked closed-loop system and the attack goal. In Section~\ref{section: AS}, we then provide  straightforward reductions from the actuator attacker synthesis problems to the  supervisor synthesis problems. In Section~\ref{section: eavesdrop}, we then consider control command eavesdropping attackers. Finally, discussions and conclusions are presented in Section~\ref{section:DC}. 
\vspace{-12pt}
\section{Preliminaries}
\label{sec: prel}
\vspace{-4pt}
We assume the reader is familiar with the basics of supervisory control theory~\cite{WMW10},~\cite{CL99} and automata theory~\cite{HU79}. In this section, we shall recall the basic notions and terminology that are necessary to understand this work. 

For any two sets $A$ and $B$, we use $A \times B$ to denote their Cartesian product and use $A-B$ to denote their set difference. $|A|$ is used to denote the cardinality of set $A$. 

A (partial) finite state automaton $G$ over alphabet  $\Sigma$ is a 5-tuple $(Q, \Sigma, \delta, q_0, Q_m)$, where $Q$ is the finite set of states, $\delta: Q \times \Sigma \longrightarrow Q$ the partial transition function\footnote{We write $\delta(q, \sigma)!$ to mean  $\delta(q, \sigma)$ is defined and write $\neg \delta(q, \sigma)!$ to mean $\delta(q, \sigma)$ is undefined.  As usual, $\delta$ is naturally extended to the partial transition function $\delta: Q \times \Sigma^* \rightarrow Q$ such that, for any $q \in Q$, any $s \in \Sigma^*$ and any $\sigma \in \Sigma$, $\delta(q, \epsilon)=q$ and $\delta(q, s\sigma)=\delta(\delta(q,s),\sigma)$. We define $\delta(Q' ,\sigma)=\{\delta(q, \sigma) \mid q \in Q'\}$ for any $Q'\subseteq Q$.}, $q_0 \in Q$ the initial state and $Q_m \subseteq Q$ the set of marked states. Let $L(G)$ and $L_m(G)$ denote the closed-behavior and the marked-behavior of $G$, respectively~\cite{WMW10}. When $Q_m=Q$, we also write $G=(Q, \Sigma, \delta, q_0)$ for simplicity. For any two finite state automata $G_1=(Q_1, \Sigma_1, \delta_1, q_{1,0}, Q_{1, m}), G_2=(Q_2, \Sigma_2, \delta_2, q_{2,0}, Q_{2, m})$, we write $G:=G_1 \lVert G_2$ to denote their synchronous product. Then, we have that
$G = (Q := Q_1\times Q_2,\Sigma:=\Sigma_1\cup\Sigma_2,\delta := {\delta}_1 \lVert {\delta}_2,q_0:=(q_{1,0},q_{2,0}), Q_m:=Q_{1, m} \times Q_{2, m})$, 
where the (partial) transition function $\delta$ is defined as follows:
for any $q = (q_1,q_2)\in Q$ and any $\sigma \in \Sigma$, 
\begin{center}$ \delta(q,\sigma):=\left\{
\begin{array}{rcl}
({\delta}_1(q_1,\sigma),q_2), && \text{if } {\sigma \in {\Sigma}_1}- {\Sigma}_2 \\
(q_1,{\delta}_2(q_2,\sigma)), && \text{if } {\sigma \in {\Sigma}_2}- {\Sigma}_1 \\
({\delta}_1(q_1,\sigma),{\delta}_2(q_2,\sigma)), && \text{if } {\sigma \in {\Sigma}_1}\cap {\Sigma}_2 \\
\end{array} \right. $
\end{center}
\noindent
A control constraint over $\Sigma$ is a tuple $(\Sigma_c, \Sigma_o)$ of sub-alphabets of $\Sigma$, where $\Sigma_o \subseteq \Sigma$ denotes the subset of observable events (for the supervisor) and $\Sigma_c \subseteq \Sigma$  denotes the subset of controllable events (for the supervisor). Let $\Sigma_{uo}=\Sigma-\Sigma_o\subseteq \Sigma$ denote the subset of unobservable events (for the supervisor) and let $\Sigma_{uc}=\Sigma-\Sigma_c \subseteq \Sigma$ denote the subset of uncontrollable events (for the supervisor). For each sub-alphabet $\Sigma' \subseteq \Sigma$, the natural projection $P_{\Sigma'}: \Sigma^* \rightarrow \Sigma'^*$ is defined, which is extended to a map between languages as usual~\cite{WMW10}. Let $G=(Q, \Sigma, \delta, q_0)$. We shall abuse the notation and define $P_{\Sigma'}(G)$ to be the finite state automaton $(2^Q, \Sigma, \Delta, UR_{G, \Sigma-\Sigma'}(q_0))$ over $\Sigma$, where the unobservable reach $UR_{G, \Sigma-\Sigma'}(q_0):=\{q \in Q \mid  \exists s \in (\Sigma-\Sigma')^*, q=\delta(q_0, s)\}$ of $q_0$ with respect to the sub-alphabet $\Sigma-\Sigma' \subseteq \Sigma$ is the initial state, and the partial transition function $\Delta: 2^Q \times \Sigma \longrightarrow 2^Q$ is defined as follows.
\begin{enumerate}
\item 
For any $\varnothing \neq Q' \subseteq Q$ and any $\sigma \in \Sigma'$, $\Delta(Q', \sigma)=UR_{G, \Sigma-\Sigma'}(\delta(Q', \sigma))$,  where $UR_{G, \Sigma-\Sigma'}(Q'):=\bigcup_{q \in Q'}UR_{G, \Sigma-\Sigma'}(q)$ for any $Q' \subseteq Q$;
\item for any $\varnothing \neq Q' \subseteq Q$ and any $\sigma \in \Sigma-\Sigma'$, $\Delta(Q', \sigma)=Q'$. 
\end{enumerate}
In particular, we shall note that $P_{\Sigma'}(G)$ is over the alphabet $\Sigma$ and there is no transition defined at the state $\varnothing \in 2^Q$. A finite state automaton $G=(Q, \Sigma, \delta, q_0, Q_m)$ is said to be non-blocking if every reachable state in $G$ can reach a marked state in $Q_m$~\cite{WMW10}. 

\vspace{-12pt}
\section{System Setup }
\label{section:SS}
\vspace{-4pt}
In this section, we shall introduce the system setup that involves non-eavesdropping actuator attacker. Intuitive explanations about the basic underlying ideas and assumptions are provided in Section~\ref{subs:BI}. After that, a formal setup is provided in Section~\ref{subs: FS}.

\vspace{-6pt}
\subsection{Basic Ideas}
\label{subs:BI}
The plant $G$, which is a finite state automaton over $\Sigma$, is under the control of a partially observing supervisor $S$ over $(\Sigma_c, \Sigma_o)$, which is also given by a finite state automaton over $\Sigma$. In addition, we assume the existence of an actuator attacker $A$  that partially observes the execution of the closed-loop system and can modify the control command $\gamma$ (issued by the supervisor) on a specified attackable subset $\Sigma_{c, A} \subseteq \Sigma_c$ of  controllable events. The plant $G$ follows the modified control command $\gamma'$ instead of $\gamma$. In this work, we assume that the supervisor issues the newly generated control command each time whenever it observes an event transition and when the system starts~\cite{LZS19}, and we assume
the plant $G$ follows the (old) modified control command $\gamma'$ whenever an event $\sigma \in \Sigma_{uo}$ unobservable to the supervisor is executed in $G$. For simplification of analysis, we shall assume $\Sigma_{o, A} \subseteq \Sigma_o$. 
We impose this assumption as we do not want the attacker to observe $\sigma$ if $\sigma \in \Sigma_{uo}$ is executed in the plant, when the plant will reuse the old modified control command $\gamma'$\footnote{This assumption may be removed, by imposing the constraint that the attacker does not change its attack decision whenever $\sigma \in \Sigma_{uo}$ is executed in the plant, even when the execution of $\sigma$ may be observed by the attacker, if $\sigma \in \Sigma_{uo} \cap \Sigma_{o, A}$, and thus may lead to state change in the attacker. The analysis of this general case is left to the future work.}. We assume the supervisor $S$ is augmented with a monitoring mechanism that monitors the execution of the closed-loop system (possibly under attack). If $\sigma \in \Sigma_{c, A}$ is enabled by the supervisor in $\gamma$, then the attacker is not discovered\footnote{The supervisor is not sure whether $\sigma$ has been disabled by an attacker, even if disabling $\sigma$ may result in deadlock, as the supervisor is never sure whether: 1) deadlock has occurred due to actuator attack, or 2) $\sigma$ will possibly fire soon (according to the internal mechanism of the plant).} by the supervisor when it disables $\sigma$ in $\gamma'$.  The supervisor is able to conclude the existence of an attacker the first moment when it observes an inconsistency between what it has observed and what shall be observed in the absence of  attacker. 
We assume that the attacker has a complete knowledge about the models of the plant $G$ and the supervisor $S$, as well as the control constraint $(\Sigma_{c}, \Sigma_{o})$ and what constitutes damages. Intuitively, the goal of the attacker is to remain covert in the process of attacking the closed-loop system until it causes damages (by generating certain damaging strings). More precisely, the covertness property requires the attacker not to reach a situation where its existence is detected by the supervisor while no damage has been or can be caused. We assume that the supervisor has a mechanism for halting the execution of the closed-loop system after discovering an actuator attack. 

\vspace{-10pt}
\subsection{Formal Setup}
\label{subs: FS}
In this section, we explain the formal setup. We  first need to introduce and present a formalization of the system components.

{\bf Supervisor}: In the absence of attacker, a supervisor over control constraint $(\Sigma_c, \Sigma_o)$ is modeled by a finite state automaton $S=(X, \Sigma, \zeta, x_0)$ that satisfies the  controllability and observability constraints~\cite{B1993}:
\begin{enumerate}
\item [$\bullet$] ({\em controllability}) for any state $x \in X$ and any uncontrollable event $\sigma \in \Sigma_{uc}$, $\zeta(x, \sigma)!$,
\item [$\bullet$] ({\em observability}) for any state $x \in X$ and any unobservable event $\sigma \in \Sigma_{uo}$, $\zeta(x, \sigma)!$ implies $\zeta(x, \sigma)=x$,
\end{enumerate}


{\bf Plant}: The plant is modeled by a finite state automaton $G=(Q, \Sigma, \delta, q_0)$ as usual. Whenever the plant fires an observable transition $\delta(q, \sigma)=q'$, it sends the observable event $\sigma$ to the supervisor. 

{\bf Actuator Attacker}: 
We assume the attacker knows the model of the plant $G$ and the model of the supervisor $S$. We here impose some restrictions on the capability of the attacker. Let $\Sigma_{o, A} \subseteq \Sigma$ denote the subset of (plant) events that can be observed by the attacker. Let $\Sigma_{c, A} \subseteq \Sigma_c$ denote the set of attackable events. The attacker is able to modify the
control command $\gamma$ issued by the supervisor on the subset $\Sigma_{c, A}$.

We henceforth name $(\Sigma_{o, A}, \Sigma_{c, A})$ as an attack constraint. An actuator attacker over attack constraint  $(\Sigma_{o, A}, \Sigma_{c, A})$ is modeled by a finite state automaton $A=(Y, \Sigma, \beta, y_0)$ that satisfies the following constraints. 
\begin{enumerate}
\item [$\bullet$] ({\em A-controllability}) for any state $y \in Y$ and any unattackable event $\sigma \in \Sigma-\Sigma_{c, A}$, $\beta(y, \sigma)!$,
\item [$\bullet$] ({\em A-observability}) for any state $y \in Y$ and any unobservable event $\sigma \in \Sigma-\Sigma_{o, A}$ to the attacker, $\beta(y, \sigma)!$ implies $\beta(y, \sigma)=y$, 
\end{enumerate}


Compared with~\cite{Su2018},~\cite{Lin2018} and~\cite{LZS19}, in this work we have adopted a unifying model for the  modeling of the plant, the supervisor and the attacker. In particular, it is straightforward to see that an attacker $A$ over attack constraint $(\Sigma_{o, A}, \Sigma_{c, A})$ can also be viewed as a supervisor over control constraint $(\Sigma_{c, A}, \Sigma_{o, A})$. We here shall remark that we do not impose any assumption  on the control constraint $(\Sigma_c, \Sigma_o)$ and we do not restrict the attack constraint $(\Sigma_{o, A}, \Sigma_{c, A})$ either, except for the natural requirement\footnote{The requirement $\Sigma_{c, A} \subseteq \Sigma_c$ is not essential, as it is only used in Remark 2, i.e., used in the discussion of when the reduction can be carried out in polynomial time.}  $\Sigma_{c, A} \subseteq \Sigma_c$   and the technical requirement\footnote{The relaxation of this requirement makes our synthesis approach  incomplete in general.} that $\Sigma_{o, A} \subseteq \Sigma_o$. In particular, we do not require $\Sigma_{o, A} = \Sigma_o$, which is required by~\cite{Carvalho2018},~\cite{LACM17} and~\cite{Lima2018}, and we also do not require normality $\Sigma_c \subseteq \Sigma_o, \Sigma_{c, A} \subseteq \Sigma_{o, A}$, which is assumed in~\cite{Lin2018}. We allow both enablement and disablement attacks; in particular, disablement attacks can be essential for the success of covert attackers, which is quite different from~\cite{Carvalho2018},~\cite{LACM17} and~\cite{Lima2018}. 

{\bf Damage Automaton}: To specify in a general manner what strings can cause damages,  we here shall adopt a complete finite state automaton $H=(W, \Sigma, \chi, w_0, W_m)$~\cite{Lin2018}, which is referred to as a damage automaton. In particular, each string $s \in L_m(H)$  is a damage-inflicting string that the attacker would like the attacked closed-loop system to generate. In the special case of enforcing or testing state avoidance property for the closed-loop system under attack, we can get rid of $H$ and introduce the set $Q_{bad} \subseteq Q$ of bad states to avoid in the plant $G$. In general, $H$ does not have to be the same as the specification that is used to synthesize the (possibly insecure) supervisor $S$. 

{\bf Monitoring Mechanism}: In the presence of an attacker, we assume the supervisor is augmented with a monitoring 
mechanism. We  assume, without loss of generality, that the execution of the attacked closed-loop system is  immediately halted\footnote{ In~\cite{LACM17}, when the supervisor detects the presence of an attacker, all controllable events will be disabled, while  uncontrollable events can still occur (i.e., immediate halt by  reset is impossible). This is not difficult to accommodate, by working on the damage automaton $H$ that defines the set of damage-inflicting strings $L_m(H)$. In particular, if uncontrollable events can still occur after the detection, then we only need to use the new damage automaton $H'$ with $L_m(H')=L_m(H) \cup L_m(H)/\Sigma_{uc}^*$, where $L_m(H)/\Sigma_{uc}^*:=\{s \in \Sigma^* \mid \exists s' \in \Sigma_{uc}^*, ss' \in L_m(H)\}$, as the renewed set of damage-inflicting strings. An alternative approach is to add self-loops of each uncontrollable event at the state $x_{halt}$ of $S^A$ and at the state $\varnothing$ of $P_{\Sigma_o}(S\lVert G)$ to allow execution of uncontrollable events after the detection; the definition of these two components will be explained soon.} once the supervisor discovers the presence of an attacker. The supervisor online records its observation $w \in \Sigma_o^*$ of the execution of the (attacked) closed-loop system. It  concludes the existence of an attacker
and then halts the execution of the (attacked) closed-loop system the first time when it observes
some string $w \notin P_{\Sigma_o}(L(S \lVert G))$\footnote{The same monitoring mechanism indeed works for both covert attackers and risky attackers~\cite{LZS19}.}. 

{\bf Attacked Closed-loop System}: The attacked closed-loop system consists of the three main components $G$, $S$ and $A$. However, the synchronous product of $G$, $S$ and $A$ is not the attacked closed-loop system for the 
following two reasons: 1) The effect of attack on the supervisor is not reflected in $S$; 2) The monitoring mechanism has not been encoded in $S$.

The first problem can be solved by computing the (totally) enablement-attacked supervisor $S^A$. Let $S=(X, \Sigma, \zeta, x_0)$.  We shall let $S^A=(X \cup \{x_{halt}\}, \Sigma, \zeta^A, x_0)$ denote the enablement-attacked supervisor, where $x_{halt} \notin X$ denotes a distinguished halting state and  $\zeta^A: (X \cup \{x_{halt}\}) \times \Sigma \longrightarrow  (X \cup \{x_{halt}\})$ is obtained from $\zeta$ by adding the following transitions (in addition to a  copy of the transitions of $\zeta$).
\begin{enumerate}
    \item [a)] for any $x \in X$ and any $\sigma \in \Sigma_{uo} \cap \Sigma_{c, A}$, if $\neg \zeta(x, \sigma)!$, then $\zeta^A(x, \sigma)=x$
    \item [b)] for any $x \in X$ and any $\sigma \in \Sigma_{o} \cap \Sigma_{c, A}$, if $\neg \zeta(x, \sigma)!$, then $\zeta^A(x, \sigma)=x_{halt}$
\end{enumerate}
Intuitively, to obtain $S^A$, for each state $x \in X$ and each event $\sigma \in \Sigma_{uo} \cap \Sigma_{c, A}$, if $\sigma$ is not defined at state $x$, we then add a self-loop at state $x$ labeled by $\sigma$; for each state $x \in X$ and each  event $\sigma \in \Sigma_o \cap \Sigma_{c, A}$, if $\sigma$ is not defined at state $x$, we then add  a transition from $x$ to the halting state $x_{halt}$ labeled by $\sigma$. In particular, we note that there is no transition defined at state $x_{halt}$ and every state in $S^A$ is a marked state. 

With $S^A$, we have effectively implemented an ignorant attacker that performs enablement attacks wherever it is possible, that is, an attacker that enables each attackable event at each state. Then, the rest of the attack decisions for the actuator attacker is to disable attackable events properly, based on its own partial observation and its knowledge on the system's models, to ensure covertness and  damage-infliction. This brings  the covert actuator attacker synthesis problem closer to the supervisor synthesis problem.

The second problem has been partially solved in the solution to the first problem. In particular, if $S^A$ enters $x_{halt}$, then an attack has necessarily happened. This is insufficient, however. Recall that the supervisor  concludes the existence of an attacker
and then halts the execution of the (attacked) closed-loop system the first time when it observes
some string $w \notin P_{\Sigma_o}(L(S \lVert G))$. We shall denote $S \lVert G=(X \times Q, \Sigma, \lambda, (x_0, q_0))$, where $\lambda=\zeta \lVert \delta$. Then, we have  
$P_{\Sigma_o}(S \lVert G)=(2^{X \times Q}, \Sigma, \Lambda, UR_{S\lVert G, \Sigma-\Sigma_{o}}(x_0, q_0))$,
where $UR_{S\lVert G, \Sigma-\Sigma_{o}}(x_0, q_0)$ denotes the unobservable reach of $(x_0, q_0)$ with respect to the sub-alphabet $\Sigma-\Sigma_o$ in $S\lVert G$ and
the partial transition function $\Lambda: 2^{X \times Q} \times \Sigma \longrightarrow 2^{X \times Q}$ is defined such that, 
\begin{enumerate}
\item 
for any $\varnothing \neq D \subseteq X \times Q$ and any $\sigma \in \Sigma_o$,  $\Lambda(D, \sigma)=UR_{S \lVert G, \Sigma-\Sigma_o}(\lambda(D, \sigma))$;
\item for any $\varnothing \neq D \subseteq X \times Q$ and any $\sigma \in \Sigma-\Sigma_o$, $\Lambda(D, \sigma)=D$. 
\end{enumerate}
Intuitively, we can use $P_{\Sigma_o}(S\lVert G)$ to track the belief $\varnothing \neq D \subseteq X \times Q$ of the supervisor on the current states of the unattacked closed-loop system $S \lVert G$, for  monitoring purpose. When the current belief of the supervisor is  $\varnothing \neq D \subseteq X \times Q$ and $\sigma \in \Sigma_o$ is executed with $\lambda(D, \sigma)=\varnothing$, the supervisor detects an inconsistency and determines that attack necessarily has  happened and thus halts the execution of the attacked closed-loop system. In particular, we note that there is no transition defined at state $\varnothing \in 2^{X \times Q}$ and every state in $P_{\Sigma_o}(S \lVert G)$ is  a marked state. 

Given the plant $G$, the supervisor $S$  and the attacker $A$, the attacked closed-loop system is the synchronous product $G \lVert S^A \lVert P_{\Sigma_o}(S \lVert G) \lVert A$, which is a finite state automaton over $\Sigma$. Since the set of states the supervisor believes itself to be residing in is a singleton (due to the observability constraint for a supervisor) and corresponds to the state that it actually resides in, whenever $P_{\Sigma_o}(S\lVert G)$ is not in the state $\varnothing$, the state space of $G \lVert S^A \lVert P_{\Sigma_o}(S \lVert G) \lVert A$, when restricted to  the reachable state set, is isomorphic to $Q \times (X \cup \{x_{halt}\}) \times 2^Q \times Y$. In the rest, we consider $Q \times (X \cup \{x_{halt}\}) \times 2^Q \times Y$ as the state space whenever convenient.  


{\bf Successful Actuator Attackers}: In this work, we shall focus on covert attackers. 
Given the damage automaton $H=(W, \Sigma, \chi, w_0, W_m)$,
the definition of  covertness is given in the following.
\begin{definition}[Covertness]
Given any plant $G$, any supervisor $S$, any actuator attacker $A$ and any damage automaton $H$, $A$ is said to be covert on $(G, S)$ with respect to $H$ if each state in $\{(q, x, D_Q, y, w) \in Q \times (X \cup \{x_{halt}\}) \times 2^Q \times Y \times W \mid (x=x_{halt} \vee D_Q=\varnothing) \wedge w \notin W_m\}$ is not reachable\footnote{Alternatively, we can use the equivalent set $\{(q, x, D, y, w) \in Q \times (X \cup \{x_{halt}\}) \times 2^{X \times Q} \times Y \times W \mid D=\varnothing \wedge w \notin W_m\}$ for specifying the violation of covertness.} in 
$G \lVert S^A \lVert P_{\Sigma_o}(S \lVert G) \lVert A \lVert H$.
\end{definition}
We note that the set of marked states for $G \lVert S^A \lVert P_{\Sigma_o}(S \lVert G) \lVert A \lVert H$  is $\{(q, x, D_Q, y, w) \in Q \times (X \cup \{x_{halt}\}) \times 2^Q \times Y \times W \mid w \in W_m\}$. In this work, we shall consider two ``possibilistic" interpretation of damaging goals for  covert actuator attackers. For us, the first interpretation of interest is the damage-nonblocking goal, which intuitively states that it is always possible to reach damages in the attacked closed-loop system, wherever the current state is for the attacked closed-loop system. 
That is, for any string $s \in L(G \lVert S^A \lVert P_{\Sigma_o}(S \lVert G) \lVert A)$ generated by the attacked closed-loop system, it is always possible to extend it in $G \lVert S^A \lVert P_{\Sigma_o}(S \lVert G) \lVert A$ to some string in $L_m(H)$. This is formally captured by the next definition.
\begin{definition}[Damage-nonblocking]
Given any plant $G$, any supervisor $S$, any damage automaton $H$ and any covert actuator attacker $A$ on $(G, S)$ with respect to $H$, $A$ is said to be damage-nonblocking on $(G, S)$ with respect to $H$ if 
$G \lVert S^A \lVert P_{\Sigma_o}(S \lVert G) \lVert A \lVert H$ is non-blocking.
\end{definition}
The damage-nonblocking goal ensures the ``necessity" of damage-infliction in the long run, much in parallel to a nonblocking supervisor. The second interpretation of interest for us is the damage-reachable goal~\cite{Lin2018},~\cite{Zhu2018},~\cite{LZS19}, which intuitively states that it is possible to reach damages in the attacked closed-loop system  from the initial state. It is worth mentioning that, due to the existence of events beyond the control of  attackers (that is, any event $\sigma \notin \Sigma_{c, A}$), this is arguably the weakest interpretation for damage-infliction goal. 
\begin{definition}[Damage-reachable]
Given any plant $G$, any supervisor $S$, any damage automaton $H$ and any covert actuator attacker $A$ on $(G, S)$ with respect to $H$,  $A$ is said to be damage-reachable on $(G, S)$ with respect to $H$ if some marked state in 
$G \lVert S^A \lVert P_{\Sigma_o}(S \lVert G) \lVert A \lVert H$ is reachable, that is, $L_m(G \lVert S^A \lVert P_{\Sigma_o}(S \lVert G) \lVert A \lVert H) \neq \varnothing$.
\end{definition}
It is clear that
the damage-nonblocking goal is stronger than the damage-reachable goal, which is given by the proposition below.
\begin{proposition}
Given any plant $G$ over $\Sigma$, any supervisor $S$ over control constraint $(\Sigma_c, \Sigma_o)$, any attack constraint $(\Sigma_{o, A}, \Sigma_{c, A})$ and any damage automaton $H$ over $\Sigma$, if a covert actuator attacker $A$ over $(\Sigma_{o, A}, \Sigma_{c, A})$ is damage-nonblocking on $(G, S)$ with respect to $H$, then $A$ is damage-reachable on $(G, S)$ with respect to $H$.
\end{proposition}
\vspace{-4pt}
\begin{proof}
This follows straightforwardly  from the definitions of the damage-nonblocking goal and the damage-reachable goal. Indeed, $G \lVert S^A \lVert P_{\Sigma_o}(S \lVert G) \lVert A \lVert H$ is non-blocking implies that some marked state in 
$G \lVert S^A \lVert P_{\Sigma_o}(S \lVert G) \lVert A \lVert H$ is reachable.
\end{proof}
However, in general, a damage-reachable attacker may not be  damage-nonblocking.
\vspace{-12pt}
\section{Synthesis of Non-eavesdropping Actuator Attackers}
\label{section: AS}
\vspace{-4pt}
In this section, we provide  straightforward reductions from the  actuator attacker synthesis problems, with the two  different damaging goals, to the  Ramadge-Wonham  supervisor synthesis problems. It then follows that we can use the  techniques and  tools developed for solving the  supervisor synthesis problems to solve the  actuator attacker synthesis problems. We shall provide the reductions for the two different interpretations of damaging goals, separately. The details of the  reductions are summarized in the next two theorems.
\begin{theorem}
Given any plant $G$ over $\Sigma$, any supervisor $S$ over control constraint $(\Sigma_c, \Sigma_o)$, any attack constraint $(\Sigma_{o, A}, \Sigma_{c, A})$ and any damage automaton $H$ over $\Sigma$, there is a damage-nonblocking covert actuator attacker $A$ over $(\Sigma_{o, A}, \Sigma_{c, A})$ on $(G, S)$ with respect to $H$ iff there is a supervisor $S'$ over $(\Sigma_{c, A}, \Sigma_{o, A})$ such that $S' \lVert P$ avoids reaching the set ${\bf BAD}$ of bad states in $P$ and $S' \lVert P$ is non-blocking, where $P=G \lVert S^A \lVert P_{\Sigma_o}(S\lVert G) \lVert H$ and ${\bf BAD}:=\{(q, x, D_Q, w) \in Q \times (X \cup \{x_{halt}\}) \times 2^Q \times W  \mid (x=x_{halt} \vee D_Q=\varnothing) \wedge w \notin W_m\}$.
\end{theorem}

\begin{proof}
This straightforwardly follows from the definition of a damage-nonblocking covert actuator attacker. In particular, $A$ is a damage-nonblocking covert actuator attacker over $(\Sigma_{o, A}, \Sigma_{c, A})$ on $(G, S)$ with respect to $H$ iff $A \lVert P$ avoids reaching the set ${\bf BAD}$ of bad states in $P$ and $A \lVert P$ is non-blocking, where $P=G \lVert S^A \lVert P_{\Sigma_o}(S \lVert G) \lVert H$ and ${\bf BAD}:=\{(q, x, D_Q, w) \in Q \times (X \cup \{x_{halt}\}) \times 2^Q \times W  \mid (x=x_{halt} \vee D_Q=\varnothing) \wedge w \notin W_m\}$, by Definition 1 and Definition 2; in particular, $A$ can be viewed as a supervisor over $(\Sigma_{c, A}, \Sigma_{o, A})$ and $P$ can be viewed as the transformed plant. 
\end{proof}

\begin{theorem}
Given any plant $G$ over $\Sigma$, any supervisor $S$ over control constraint $(\Sigma_c, \Sigma_o)$, any attack constraint $(\Sigma_{o, A}, \Sigma_{c, A})$ and any damage automaton $H$ over $\Sigma$, there is a damage-reachable covert actuator attacker $A$ over $(\Sigma_{o, A}, \Sigma_{c, A})$ on $(G, S)$ with respect to $H$ iff there is a supervisor $S'$ over $(\Sigma_{c, A}, \Sigma_{o, A})$ such that $S' \lVert P$ avoids reaching the set ${\bf BAD}$ of bad states in $P$ and $L_m(S' \lVert P) \neq \varnothing$, where $P=G \lVert S^A \lVert P_{\Sigma_o}(S \lVert G) \lVert H$ and ${\bf BAD}:=\{(q, x, D_Q, w) \in Q \times (X \cup \{x_{halt}\}) \times 2^Q \times W  \mid (x=x_{halt} \vee D_Q=\varnothing) \wedge w \notin W_m\}$.
\end{theorem}

\begin{proof}
This straightforwardly follows from the definition of a damage-reachable covert actuator attacker. In particular, $A$ is a damage-reachable covert actuator attacker over the attack constraint $(\Sigma_{o, A}, \Sigma_{c, A})$ on $(G, S)$ with respect to $H$ iff $A \lVert P$ avoids reaching the set ${\bf BAD}$ of bad states in $P$ and $L_m(A \lVert P) \neq \varnothing$, where $P=G \lVert S^A \lVert P_{\Sigma_o}(S \lVert G) \lVert H$ and ${\bf BAD}:=\{(q, x, D_Q, w) \in Q \times (X \cup \{x_{halt}\}) \times 2^Q \times W  \mid (x=x_{halt} \vee D_Q=\varnothing) \wedge w \notin W_m\}$, by Definition 1 and Definition 3; in particular, $A$ can be viewed as a supervisor over the control constraint $(\Sigma_{c, A}, \Sigma_{o, A})$ and $P$ can be viewed as the transformed plant. 
\end{proof}
\vspace{-4pt}
\begin{remark}
By viewing $A$ as the supervisor and $G\lVert S^A \lVert P_{\Sigma_o}(S \lVert G)  \lVert H$ as the transformed plant, Definition 2 and Definition 3 allow us to reduce the actuator attacker synthesis problems to the supervisor synthesis problems. 
We remark that the set {\bf BAD} of bad states in Theorem 1 and Theorem 2 is defined according to the definition of covertness in Definition 1, where breaking covertness, i.e.,  reaching states in {\bf BAD}, is bad for the attacker. 
\end{remark}
The reductions provided above,  combined with the maximally permissive supervisor synthesis algorithm\footnote{Theorem 1 requires the synthesis of a safe and nonblocking supervisor, while Theorem 2 only requires the synthesis of a safe and nonempty supervisor. The  maximally permissive supervisor synthesis algorithm in~\cite{YL16} is of exponential time complexity and   can be adapted for both scenarios.}  in~\cite{YL16},  provides a  doubly exponential time algorithm for the covert actuator attacker synthesis problem, for both the  damage-nonblocking goal and the damage-reachable goal. This doubly exponential time complexity seems unavoidable in general, due to two independent sources of exponentiation constructions arising from partial observation. The inner level of the exponentiation is due to partial observation by the supervisor, in the construction of $P_{\Sigma_o}(S \lVert G)$ to determine when the attacked closed-loop system is halted. The outer level of the exponentiation is due to partial observation by the attacker, whose synthesis can be viewed as the synthesis of a partial observation supervisor on the transformed plant. In general, the reductions provided in the above two theorems cannot be carried out in polynomial time. Here, the main trouble is due to the need for tracking the state in $P_{\Sigma_o}(S \lVert G)$, whose state size is in general exponential in the state size of $G$ (modulo the state space of $S$). We now shall show that, under certain  assumption on the  set of attackable events, the above reductions are polynomial time computable, regardless of the models for $S$ and $G$.

\begin{theorem}
If $\Sigma_{uo} \cap \Sigma_{c, A} =\varnothing$, then the actuator attacker synthesis problem is polynomial time reducible to the (partial-observation) supervisor synthesis problem, for both the damage-nonblocking goal and the damage-reachable goal.
\end{theorem}

\begin{proof}
Consider the transformed plant $G \lVert S^A \lVert P_{\Sigma_o}(S \lVert G) \lVert H$ in Theorem 1 and Theorem 2. We  restore the component $2^{X \times Q}$ for analysis; then a state of $G \lVert S^A \lVert P_{\Sigma_o}(S \lVert G) \lVert H$ is of the form $(q, x, D, w) \in Q \times (X \cup \{x_{halt}\}) \times 2^{X \times Q} \times W$. Let $dom(D):=\{x \in X \mid \exists q \in Q, (x, q) \in D\}$. We know that $dom(D)$ is the singleton $\{x\}$, when restricted to the reachable state set of $G \lVert S^A \lVert P_{\Sigma_o}(S \lVert G) \lVert H$ and when $D \neq \varnothing$. Let $Img(D):=\{q \in D \mid \exists x \in X, (x, q) \in D\}$. 

When the belief of the supervisor is  $\varnothing \neq D \subseteq X \times Q$ and event  $\sigma \in \Sigma_o$ is executed with $\lambda(D, \sigma)=\varnothing$, the supervisor detects an inconsistency and determines that attack has necessarily happened. There are two possible reasons for $\lambda(D, \sigma)=\varnothing$. 
The first reason is when $\zeta(x, \sigma)$ is undefined for the  unique state $x \in dom(D)$. This has already been taken care of by the component $S^A$. The second reason is when $\delta(Img(D), \sigma)=\varnothing$, which can be tracked with the component $P_{\Sigma_o}(S \lVert G)$. However, since  $\sigma$ has been executed at the current plant state $q \in Q$, we conclude that $q \notin Img(D)$. That is, the current state the plant is in does not belong to the set of states the supervisor believes the plant to be residing in. This is only possible if some event $\sigma' \in \Sigma_{uo} \cap \Sigma_{c, A}$ has been enabled by the actuator attacker but disabled by the supervisor in the past execution. If $\Sigma_{uo} \cap \Sigma_{c, A}=\varnothing$, then the above possibility can be ruled out. 

Thus, we can  effectively get rid of the component $2^Q$ (associated with the use of $P_{\Sigma_o}(S \lVert G)$) in the construction of the transformed plant, when
$\Sigma_{uo} \cap \Sigma_{c, A}=\varnothing$, since the halting state $x_{halt}$ in the component $S^A$ is sufficient for detecting the presence of an actuator attacker in this case. It then follows that the reductions provided in Theorems 1 and 2 are polynomial time reductions when $\Sigma_{uo} \cap \Sigma_{c, A}=\varnothing$. 
\end{proof}
Intuitively, if the actuator attacker cannot attack those unobservable events to the supervisor, then the reductions provided in Theorem 1 and Theorem 2, after removing the component $2^Q$, are polynomial time reductions.
In the following remark, we shall explain one implication of Theorem 3.
\begin{remark}
\label{remark: poly}
If $\Sigma_{uo} \cap \Sigma_c=\varnothing$,
that is, when the given supervisor $S$ satisfies the normality property~\cite{WMW10}, then we also have $\Sigma_{uo} \cap \Sigma_{c, A}=\varnothing$, since $\Sigma_{c, A} \subseteq \Sigma_c$. We consider this to be not quite restrictive. 
A supervisor satisfying the normality property leaves much reduced attack surfaces for the actuator attacker~\cite{Lin2018},
since events that are disabled (by the supervisor) are observable to the
supervisor and enabling those disabled (attackable) events (by the actuator attacker)
may immediately break the covertness of the actuator attacker, if they are also enabled by the
plant. Thus, being suspicious of the presence of an attacker, the designer may prefer to synthesize supervisors that satisfy the normality property. 
\end{remark}
Based on Theorem 3 and the discussion after Remark 1, we immediately have the following result.
\begin{corollary}
If $\Sigma_{uo} \cap \Sigma_{c, A} =\varnothing$, then the actuator attacker synthesis problem is solvable in exponential time, for both the damage-nonblocking goal and the damage-reachable goal.
\end{corollary}

In general, there are two possible situations where the exponential blowup of state sizes can be avoided in the construction of the transformed plant $G \lVert S^A \lVert P_{\Sigma_o}(S \lVert G) \lVert H$. The first situation is when $P_{\Sigma_o}(S \lVert G)$ is of polynomial size. An example is when  $P_{\Sigma_o}$ is an $L(S\lVert G)$-observer~\cite{LW10}; in this case, the size of $P_{\Sigma_o}(S \lVert G)$ is upper bounded by the state size of $S \lVert G$. The second situation is when for any reachable state $(q, x, D, w) \in Q \times (X \cup \{x_{halt}\}) \times 2^{X \times Q} \times W$ and  any $\sigma \in \Sigma_o$ such that $\sigma$ is defined at $(q, x, D, w)$, we have  $\delta(Img(D), \sigma) \neq \varnothing$; in this case, the component $P_{\Sigma_o}(S \lVert G)$ can be safely removed in the construction of the transformed plant, as analyzed in the proof of Theorem 3. An example is when $\Sigma_{uo} \cap \Sigma_{c, A}=\varnothing$, as shown in Theorem 3. For the second situation, in general we need to first construct $G \lVert S^A \lVert P_{\Sigma_o}(S \lVert G) \lVert H$ and then perform the corresponding (straightforward) verification. If the condition is indeed satisfied, then we can safely remove the component $P_{\Sigma_o}(S \lVert G)$ and perform supervisor synthesis based on $G \lVert S^A \lVert H$. In both situations, the actuator attacker synthesis problem is solvable in  exponential time, for both the damage-nonblocking goal and the damage-reachable goal. We here remark that the above two situations are not exclusive. In general, it is  possible that $P_{\Sigma_o}(S \lVert G)$ is of polynomial size and still can be removed in the construction of the transformed plant (according to the second situation).

\section{ Synthesis of Eavesdropping Actuator Attackers}
\label{section: eavesdrop}
\vspace{-4pt}
In this section, we  consider the synthesis of successful covert actuator attackers that not only  partially observes the  execution of the closed-loop system  but also can eavesdrop the control commands issued by the supervisor, following the setup  of~\cite{Lin2018}. As we have explained before, control command eavesdropping attackers can be considered as powerful attackers in the sense of~\cite{HMR18}; in particular, eavesdropping attackers can decode the information encoded in the issued control commands to possibly gain part of the knowledge of the supervisors on the execution of the plant. 

Since the attackers can eavesdrop the control commands, it is important to discuss when the control commands are sent to the plant. The control command generated at each supervisor state $x \in X$ is  $\Gamma(x):=\{\sigma \in \Sigma \mid \zeta(x, \sigma)!\}$. We assume that whenever the supervisor fires an observable transition $\zeta(x, \sigma)=x'$, it sends the newly generated control command $\Gamma(x')$ to the plant\footnote{In another setup~\cite{LZS19}, the supervisor sends a control command each time when it fires an observable transition $\zeta(x, \sigma)=x'$ satisfying $\Gamma(x) \neq \Gamma(x')$. To solve the (covert) actuator attacker synthesis problem for this setup requires the concept of dynamic observation~\cite{AVW03} that is beyond the scope of this work.}. When the system first initiates, the supervisor sends the initial control command $\Gamma(x_0)$ to the plant. This setup is a generalization of that of~\cite{Lin2018} and we do not impose any assumption on the control constraint $(\Sigma_c, \Sigma_o)$ or the attack constraint $(\Sigma_{o, A}, \Sigma_{c, A})$, except for the natural requirement that $\Sigma_{c, A} \subseteq \Sigma_c$  and the technical requirement that $\Sigma_{o, A} \subseteq \Sigma_o$. In particular, we do not require normality $\Sigma_c \subseteq \Sigma_o, \Sigma_{c, A} \subseteq \Sigma_{o, A}$  in~\cite{Lin2018}.

We here shall employ a  technique that we refer to as supervisor bipartization to solve the (covert) actuator attacker synthesis problem for this setup.
Intuitively, we need to perform a (bipartization) transformation of the supervisor so that the issued control commands  become a part of the alphabet that can be observed by the attacker, while the control function of the bipartite supervisor (on the plant) remains the same as the (ordinary) supervisor. Given any supervisor $S=(X, \Sigma, \zeta, x_0)$, its bipartization is given by the finite state automaton $BT(S)=(X \cup X_{com}, \Sigma \cup \Gamma, \zeta^{BT}, x_{0,com})$, where $X_{com}=\{x_{com} \mid x \in X\}$ is a (relabelled) copy of $X$ with $X \cap X_{com}=\varnothing$, $x_{0, com} \in X_{com}$ is a copy of $x_0 \in X$ and is the initial state of $BT(S)$, and
the partial transition function $\zeta^{BT}: (X \cup X_{com}) \times (\Sigma \cup \Gamma) \longrightarrow (X \cup X_{com})$ is defined as follows:
\begin{enumerate}
    \item for any $x_{com} \in X_{com}$, $\zeta^{BT}(x_{com},\Gamma(x))=x$
    \item for any $x \in X$ and any $\sigma \in \Sigma_{uo}$, $\zeta^{BT}(x, \sigma)=\zeta(x,\sigma)$
    \item for any $x \in X$ and any $\sigma \in \Sigma_{o}$, $\zeta^{BT}(x, \sigma)=\zeta(x,\sigma)_{com}$ 
\end{enumerate}
Intuitively, each $x_{com}$ is the control state  corresponding to $x$ that is ready to issue the control command $\Gamma(x)$; the transition $\zeta^{BT}(x_{com},\Gamma(x))=x$ represents the event of the supervisor sending the control command $\Gamma(x)$ to the plant. Each $x \in X$ is a reaction state that is ready to react to an event in the issued control command $\Gamma(x)$ executed by the plant $G$. For any $x \in X$ and $\sigma \in \Sigma$, the supervisor reacts to the corresponding transition fired in the plant. If $\sigma \in \Sigma_{uo}$ is defined at $x \in X$, then the supervisor observes nothing and remains in the same reaction state $\zeta(x,\sigma)=x \in X$; if $\sigma \in \Sigma_o$ is defined at $x \in X$, then the supervisor proceeds to the next control state $\zeta(x,\sigma)_{com}$. We remark that $\zeta^{BT}(x, \sigma)$ is defined iff $\zeta(x, \sigma)$ is defined, for any $x\in X$ and any $\sigma \in \Sigma$. If we abstract $BT(S)$ by merging the states $x_{com}$ and $x$, treated as equivalent states in the abstraction, then we can recover  $S$. In this sense, $BT(S)$ is control equivalent to $S$. We shall refer to $BT(S)$ as a bipartite supervisor over $\Sigma \cup \Gamma$.

A control command eavesdropping actuator attacker over attack constraint $(\Sigma_{o, A}, \Sigma_{c, A})$ is modeled by an (ordinary) actuator attacker $A=(Y, \Sigma \cup \Gamma, \beta, y_0)$ over attack constraint $(\Sigma_{o, A} \cup \Gamma, \Sigma_{c, A})$. Thus, it satisfies the following constraints.
\begin{enumerate}
    \item [$\bullet$] ({\em A-controllability}) for any state $y \in Y$ and any unattackable event $\sigma \in \Sigma \cup \Gamma-\Sigma_{c, A}$, $\beta(y, \sigma)!$,
\item [$\bullet$] ({\em A-observability}) for any state $y \in Y$ and any unobservable event $\sigma \in (\Sigma\cup \Gamma)-(\Sigma_{o, A} \cup \Gamma)=\Sigma-\Sigma_{o, A}$ to the attacker, $\beta(y, \sigma)!$ implies $\beta(y, \sigma)=y$,
\end{enumerate}
The monitoring mechanism is exactly the same as that in the non-eavesdropping case. For the bipartite supervisor $BT(S)$, we still need to compute the (totally) enablement-attacked bipartite supervisor $BT(S)^A=(X \cup X_{com} \cup \{x_{halt}\}, \Sigma \cup \Gamma, \zeta^{BT, A}, x_{0,com})$, where $x_{halt} \notin X \cup X_{com}$ is the distinguished halting state and the partial transition function $\zeta^{BT, A}: (X \cup X_{com} \cup \{x_{halt}\}) \times (\Sigma \cup \Gamma) \longrightarrow (X \cup X_{com} \cup \{x_{halt}\})$ is obtained from $\zeta^{BT}$ by adding the following transitions (in addition to a copy of the transitions of $\zeta^{BT}$). 
\begin{enumerate}
    \item [a)] for any $x \in X$ and any $\sigma \in \Sigma_{uo} \cap \Sigma_{c, A}$, if $\neg \zeta^{BT}(x, \sigma)!$, then $\zeta^{BT,A}(x, \sigma)=x$
    \item [b)] for any $x \in X$ and any $\sigma \in \Sigma_{o} \cap \Sigma_{c, A}$, if $\neg \zeta^{BT}(x, \sigma)!$, then $\zeta^{BT,A}(x, \sigma)=x_{halt}$
\end{enumerate}
 We here shall remark that, for bipartite supervisors, no modification shall be made to the control states $x \in X_{com}$. The attacked closed-loop system is $G\lVert BT(S)^A \lVert P_{\Sigma_o}(S \lVert G) \lVert A$. Finally, given the damage automaton $H=(W, \Sigma, \chi, w_0, W_m)$, we need to compute the synchronous product $G\lVert BT(S)^A \lVert P_{\Sigma_o}(S \lVert G) \lVert A \lVert H$ to determine whether the attacker $A$ is successful, with respect to the damage-nonblocking goal or the damage-reachable goal. Intuitively, the attacker $A$ can be viewed as a supervisor over control constraint $(\Sigma_{c, A}, \Sigma_{o, A} \cup \Gamma)$ and $G\lVert BT(S)^A \lVert P_{\Sigma_o}(S \lVert G) \lVert H$ can be viewed as the transformed plant over $\Sigma \cup \Gamma$. By a similar reasoning as before, we have the following theorems. 
\begin{theorem}
Given any plant $G$ over $\Sigma$, any supervisor $S$ over control constraint $(\Sigma_c, \Sigma_o)$, any attack constraint $(\Sigma_{o, A}, \Sigma_{c, A})$ and any damage automaton $H$ over $\Sigma$, there is a damage-nonblocking covert control command eavesdropping actuator attacker $A$ over $(\Sigma_{o, A}, \Sigma_{c, A})$ on $(G, S)$ with respect to $H$ iff there is a supervisor $S'$ over $(\Sigma_{c, A}, \Sigma_{o, A} \cup \Gamma)$ such that $S' \lVert P$ avoids reaching the set ${\bf BAD}$ of bad states in $P$ and $S' \lVert P$ is non-blocking, where $P=G \lVert BT(S)^A \lVert P_{\Sigma_o}(S\lVert G) \lVert H$ and ${\bf BAD}:=\{(q, x, D_Q, w) \in Q \times (X \cup X_{com} \cup \{x_{halt}\}) \times 2^Q \times W  \mid (x=x_{halt} \vee D_Q=\varnothing) \wedge w \notin W_m\}$.
\end{theorem}
\begin{theorem}
Given any plant $G$ over $\Sigma$, any supervisor $S$ over control constraint $(\Sigma_c, \Sigma_o)$, any attack constraint $(\Sigma_{o, A}, \Sigma_{c, A})$ and any damage automaton $H$ over $\Sigma$, there is a damage-reachable covert control command eavesdropping actuator attacker $A$ over attack constraint $(\Sigma_{o, A}, \Sigma_{c, A})$ on $(G, S)$ with respect to $H$ iff there is a supervisor $S'$ over $(\Sigma_{c, A}, \Sigma_{o, A} \cup \Gamma)$ such that $S' \lVert P$ avoids reaching the set ${\bf BAD}$ of bad states in $P$ and $L_m(S' \lVert P) \neq \varnothing$, where $P=G \lVert BT(S)^A \lVert P_{\Sigma_o}(S \lVert G) \lVert H$ and ${\bf BAD}:=\{(q, x, D_Q, w) \in Q \times (X \cup X_{com} \cup \{x_{halt}\}) \times 2^Q \times W  \mid (x=x_{halt} \vee D_Q=\varnothing) \wedge w \notin W_m\}$.
\end{theorem}
The same conclusions on the complexity issues as in the non-eavesdropping case, which are presented in Section 4, apply to the control command eavesdropping case. Due to serious space restriction, we here shall omit the discussions. 
\begin{remark}
The approach that uses bipartization of supervisors is rather versatile~\cite{ZLWS19}. 
Here, we shall explain how the supervisor bipartization approach can also be used for solving the  actuator attacker synthesis problem in Section 4 (thus treating the analysis in Section 4 as a special case of the analysis carried out in Section 5). Indeed, for the non-eavesdropping case, we only need to treat $\Gamma$ as unobservable to the attacker. By merging the states $x$ and $x_{com}$ in $BT(S)^A$, as the only transition between them is both unobservable and uncontrollable to the attacker, we  immediately recover $S^A$; thus, Theorem 4 and Theorem 5 are now immediately reduced to Theorem 1 and Theorem 2, respectively, for the non-eavesdropping case. The case when the supervisor sends a control command each time when it fires an observable transition $\zeta(x, \sigma)=x'$ satisfying $\Gamma(x) \neq \Gamma(x')$ can be dealt with using exactly the same construction, but requires the notion of dynamic observation~\cite{AVW03} that is beyond the scope of this work.
\end{remark}

\vspace{-12pt}
\section{Discussions and Conclusions}
\label{section:DC}
\vspace{-4pt}
In this paper, we have presented a formalization of the covert actuator attacker synthesis problem, in the formalism of discrete-event systems, for cyber-physical systems. In particular, we have addressed the problems of synthesis of covert actuator attackers with the damage-nonblocking goal and with the damage-reachable goal, respectively. It is worth noting that the technique to solve the actuator attacker synthesis problem is by reduction to the well-studied supervisor synthesis problem, instead of developing new synthesis algorithms. In the rest of this section, we provide a discussion of some future works that can be carried out. 

{\bf Synthesis of (Powerful) Actuator and Sensor Attackers}: This work only considers the synthesis of actuator attackers. One of the immediate future works for us is to consider the problem of  synthesis of (powerful) combined actuator and sensor attackers~\cite{LZS19}. The solution to the attacker synthesis problem has immediate applications to the resilient supervisor synthesis problem, as explained in~\cite{LZS19}. 


{\bf Resilient Supervisor Synthesis}: The problem of resilient supervisor synthesis can be better understood.  In~\cite{LZS19}, a bounded resilient supervisor synthesis approach has been developed as a semi-decision procedure for  the synthesis of resilient supervisors against powerful actuator and sensor attackers. It is also of interest for us to explore  other (semi-)decision procedures for resilient supervisor synthesis.

Finally, we shall remark that the problem of covertly attacking  closed-loop systems in ``continuous" control theory is nothing new~\cite{S11},~\cite{TSSJ12},~\cite{ALSB10},~\cite{PDB13},~\cite{MKBDLPS12}. It is of interest to borrow the ideas from theses studies and enrich our discrete-event systems based theory; we hope that the discrete-event based approach can also inspire further research in ``continous" control theory based security of cyber-physical systems, e.g., \begin{samepage}
 the notion and synthesis of a supremal attacker~\cite{Lin2018}. An interesting research direction is to extend our framework to hybrid or probabilistic models~\cite{FKNP11}. 

\begin{acknowledgements}
This work is financially supported by Singapore Ministry of Education Academic Research Grant RG91/18-(S)-SU RONG (VP), which is gratefully acknowledged. 
\end{acknowledgements}
\vspace{-4pt}

\addtolength{\textheight}{-10cm}   

\end{samepage}

\end{document}